\documentclass[12pt]{article}
\usepackage{amsmath,amssymb,amsthm,mathrsfs,url,hyperref}

\title{Energy--Time Uncertainty Relation for\\ Absorbing Boundaries}

\author{
Roderich Tumulka\footnote{Fachbereich Mathematik, Eberhard-Karls-Universit\"at, 
	Auf der Morgenstelle 10, 72070 T\"ubingen, Germany. 
	E-mail: roderich.tumulka@uni-tuebingen.de}
}

\date{August 29, 2022}

\newcommand{\Hilbert}{\mathscr{H}}
\newcommand{\be}{\begin{equation}}
\newcommand{\ee}{\end{equation}}
\newcommand{\scp}[2]{\langle #1|#2\rangle}

\renewcommand{\Im}{\mathrm{Im}}
\newcommand{\EEE}{\mathbb{E}}
\newcommand{\RRR}{\mathbb{R}}
\newcommand{\SSS}{\mathbb{S}}

\newcommand{\vx}{\boldsymbol{x}}

\newcommand{\vn}{\boldsymbol{n}}

\newcommand{\vX}{{\boldsymbol{X}}}
\newcommand{\vj}{{\boldsymbol{j}}}

\newcommand{\vomega}{{\boldsymbol{\omega}}}
\newcommand{\R}{\Omega}
\newcommand{\bou}{\partial \R}
\newcommand{\prob}{\mathrm{Prob}}
\newcommand{\free}{\mathrm{free}}
\newcommand{\Dir}{\mathrm{Dir}}
\theoremstyle{theorem}
\newtheorem{thm}{Theorem}

\begin{document}
\maketitle
\begin{abstract}
We prove the uncertainty relation $\sigma_T \, \sigma_E \geq \hbar/2$ between the time $T$ of detection of a quantum particle on the surface $\partial \Omega$ of a region $\Omega\subset \mathbb{R}^3$ containing the particle's initial wave function, using the ``absorbing boundary rule'' for detection time, and the energy $E$ of the initial wave function. Here, $\sigma$ denotes the standard deviation of the probability distribution associated with a quantum observable and a wave function. Since $T$ is associated with a POVM rather than a self-adjoint operator, the relation is not an instance of the standard version of the uncertainty relation due to Robertson and Schr\"odinger. We also prove that if there is nonzero probability that the particle never reaches $\partial \Omega$ (in which case we write $T=\infty$), and if $\sigma_T$ denotes the standard deviation conditional on the event $T<\infty$, then $\sigma_T \, \sigma_E \geq (\hbar/2) \sqrt{\prob(T<\infty)}$.

\medskip

\end{abstract}

\section{Introduction}

{\it Note added.} After the completion of this paper, an anonymous referee pointed out to me that the result I prove here is a special case of a result proved before in the appendix of \cite{KRSW12}. Although I had cited \cite{KRSW12}, I had missed that result, partly because it is only in the appendix but not stated in the main body of \cite{KRSW12}. If I had been aware of what is in the appendix of \cite{KRSW12}, I would not have written this paper. Still, this paper may be useful to some readers.

\bigskip

The Heisenberg uncertainty relation is perhaps best known in the form
\be\label{QP}
\sigma_{Q,\psi}\, \sigma_{P,\psi} \geq \frac{\hbar}{2}
\ee
due to Kennard \cite{Ken27}, where $\sigma_{Q,\psi}$ and $\sigma_{P,\psi}$ are the standard deviations of the probability distributions of position and momentum associated with a given wave function $\psi$ of norm 1. A well-known generalization to arbitrary self-adjoint operators $A,B$ in the place of $Q,P$, established by Robertson \cite{Rob} and Schr\"odinger \cite{Schr}, asserts that
\be\label{AB}
\sigma_{A,\psi} \,\sigma_{B,\psi} \geq \tfrac{1}{2} \Bigl| \scp{\psi}{[A,B]|\psi} \Bigr|\,.
\ee
It was expected already in the early days of quantum mechanics that a relation of the form
\be\label{TE}
\sigma_{T,\psi}\, \sigma_{E,\psi}\geq \frac{\hbar}{2}
\ee
should hold between time $T$ and energy $E$; however, it is far from obvious which ``time observable'' should be meant here, and various possibilities have been discussed in the literature, see, e.g., \cite{AB61,Busch08,KRSW12,DD15} and the references therein.

In Section~\ref{sec:pfTE}, we prove \eqref{TE} for a particular time observable, the time at which a non-relativistic quantum particle, whose initial wave function $\psi_0$ is concentrated in some region $\R\subset \RRR^3$ of physical space, gets registered by ideal detectors placed along the surface $\bou$ of $\R$. Mathematically, we take this observable to be defined by the \emph{absorbing boundary rule} \cite{Wer87,detect-rule,detect-several,detect-dirac,detect-imaginary,detect-thm}, which asserts that, in the presence of such detectors, $\psi$ evolves for $t\geq 0$ according to the Schr\"odinger equation
\be\label{Schr}
i\hbar\frac{\partial \psi}{\partial t} = -\frac{\hbar^2}{2m} \nabla^2 \psi
\ee
with the absorbing boundary condition
\be\label{abc}
\vn(\vx) \cdot \nabla \psi(\vx) = i\kappa\psi(\vx)
\ee
for every $\vx\in\bou$, where $\vn(\vx)$ is the outward unit normal vector to $\bou$ at $\vx$ and $\kappa>0$ is a constant (the wave number of sensitivity of the detectors). The rule asserts further that the probability that the time $T$ and location $\vX$ of detection lie in the interval $[t,t+dt]$ and the surface element $d^2\vx$ in $\bou$ around $\vx$, respectively, is given by
\begin{align}
\prob_{\psi_0}\bigl(T\in dt, \vX\in d^2\vx\bigr) 
&= \vn(\vx) \cdot \vj^{\psi_t}(\vx) \,dt\, d^2\vx \label{prob1}\\
&= \tfrac{\hbar\kappa}{m} |\psi_t(\vx)|^2 \, dt\,d^2\vx\,,\label{prob2}
\end{align}
where it is assumed that $\|\psi_0\|=1$, $\vj^{\psi}$ is the usual probability current associated with $\psi$,
\be\label{jdef}
\vj^{\psi} = \tfrac{\hbar}{m} \, \Im[\psi^* \nabla \psi]\,,
\ee
and \eqref{prob1} and \eqref{prob2} are equivalent by virtue of the boundary condition \eqref{abc}. 

The absorbing boundary rule represents an ideal ``hard'' detector, i.e., one that detects the particle as soon as it reaches $\bou$, as opposed to a ``soft'' detector that may take a while to notice a particle moving through the detector volume. Soft detectors can be conveniently modeled through imaginary potentials \cite{All69b,detect-imaginary}; an energy--time uncertainty relation concerning the time at which a soft detector clicks was established by Kiukas et al.~\cite{KRSW12}. That result, together with the facts that the relation does not depend on the parameters of the soft detector and that the absorbing boundary rule can be obtained as a limit of soft detectors \cite{detect-imaginary}, could also provide a strategy for proving \eqref{TE} for hard detectors; however, this strategy is not straightforward because the relevant limit involves changing the detector volume and thus the Hilbert space. Be that as it may, we will give a different, direct proof of \eqref{TE} for the absorbing boundary rule.

In the event that the particle is never detected, we write $T=\infty$. The probability
\be
p:= \prob_{\psi_0}(T<\infty)
\ee 
is 1 if $\R$ is a bounded set, but may otherwise be less than 1, for example when $\R$ is a half space. If $0<p<1$ then $T$ necessarily has infinite mean and variance, so \eqref{TE} is trivially true because $\sigma_{T,\psi}=\infty$. In such a scenario it may be of interest (as suggested in \cite{KRSW12}) to consider the conditional distribution of $T$, given that $T<\infty$, and define $\sigma_{T,\psi}$ as the standard deviation of \emph{that} distribution; for this case we will show in Section~\ref{sec:pfTEp} that
\be\label{TEp}
\sigma_{T,\psi} \, \sigma_{E,\psi} \geq \sqrt{p}\frac{\hbar}{2}
\ee
instead of \eqref{TE}. 

It is noteworthy that the constant in \eqref{TE} does not depend on the detector parameter $\kappa$, and that the constant in \eqref{TEp} depends on it only through $p$. The question thus arises whether the constants in \eqref{TE} and \eqref{TEp} are sharp; at present, I cannot answer this question. (I give some discussion at the end of Section~\ref{sec:pfTE}.)

It would also be of interest to study whether the relation 
\be
\EEE_\psi\bigl[ T |T<\infty\bigr] \: \sigma_{E,\psi} \geq C \, \sqrt{p} \, \hbar
\ee
with $\EEE_\psi[ T |T<\infty]$ the expectation of $T$ conditional on $T<\infty$, proved in \cite{KRSW12} for soft detectors, is also true for the absorbing boundary rule, and if so, what the optimal value of the numerical constant $C$ is. Likewise, relations between the means of $T$ and $E$, i.e., $\EEE_\psi[ T |T<\infty]$ and $\EEE_\psi E$, would be of interest. These questions will not be addressed here.

\section{Definition of Energy}

We need to say more about the exact statement involving \eqref{TE}, in particular what exactly the quantities in \eqref{TE} mean for us. As mentioned, the quantity $\sigma_{T,\psi}$ in \eqref{TE} is the standard deviation of the probability distribution of $T$, 
\be\label{sigmaT}
\sigma_{T,\psi}^2 = \mathrm{Var}(T)=\EEE_\psi \bigl[ (T-\EEE_\psi T)^2 \bigr]\,.
\ee
We assume that $\psi=\psi_0$ is smooth with compact support that does not touch $\bou$; for such $\psi$ we write $\psi\in C_0^\infty(\R\setminus\bou)$. For the proof to go through, however, it suffices to make the weaker assumption that $\psi$ lies in the domain $D(H^2)$ of $H^2$ and is such that both $\psi$ and $H\psi$ vanish on the boundary.

It follows \cite{detect-rule,detect-thm} from \eqref{prob2} that the probability distribution of $T$ is given by a POVM (positive-operator-valued measure \cite{Nai40,Lud85,povm}) $F(\cdot)$ in the sense that
\be
\prob_{\psi_0}(T\in\Delta) = \scp{\psi_0}{F(\Delta)|\psi_0}
\ee
for every (measurable) set $\Delta \subseteq [0,\infty]$. Since the distribution of $T$ is given by a POVM, not a self-adjoint operator, the Robertson--Schr\"odinger inequality \eqref{AB} does not directly apply to yield the desired relation \eqref{TE}. The POVM can be expressed as
\begin{align}
  F \bigl( dt \bigr) 
  &= \tfrac{\hbar\kappa}{m} \, W_t^* \biggl(\int_{\bou} d^2\vx \,
  |\vx\rangle\langle\vx| \biggr) W_t\,  dt \\
  &= -\tfrac{i}{\hbar} W_t^*(H^*-H)W_t \, dt\\[2mm]
F (\{\infty\}) 
&=\lim_{t\to\infty} W_t^* W_t 
\end{align}
with ${}^*$ denoting the adjoint operator and $W_t$ the (non-unitary) linear operator that maps $\psi_0$ to $\psi_t$ solving \eqref{Schr} and \eqref{abc}. The operators $W_t$ are of the form $W_t=e^{-iHt/\hbar}$ for $t\geq 0$ but the Hamiltonian $H$ is \emph{not} self-adjoint. The $W_t$ have the properties $W_0=I$, $W_t W_s=W_{t+s}$, and $\|W_t\psi\|\leq \|\psi\|$; that is, they form a \emph{contraction semigroup}. 

This brings us to the question of what exactly is meant by $\sigma_{E,\psi}$. The question arises because the Hamiltonian $H$ is not self-adjoint and, in fact, not unitarily diagonalizable (it is not a ``normal'' operator, its self-adjoint and its skew-adjoint part do not commute and therefore cannot be simultaneously diagonalized \cite{detect-thm}). As a consequence, there is no PVM (projection-valued measure) or POVM associated with $H$, so it is far from clear whether and how a probability distribution (over real or complex ``energies'') can be associated with $H$ and $\psi=\psi_0$. We will use instead the free Hamiltonian
\be
H_{\free}=-\frac{\hbar^2}{2m}\nabla^2
\ee
on (the appropriate domain in) $L^2(\RRR^3)$ (i.e., the second Sobolev space). This is anyhow what one would naturally regard as the ``energy distribution'' of $\psi\in C_0^\infty(\R\setminus\bou)$. The free Hamiltonian is self-adjoint and thus associated with a PVM that defines, for every $\psi\in L^2(\RRR^3)$, a probability distribution on the energy axis. In fact, this distribution is concentrated on the positive half axis and there has density 
\be
\rho(E)=
\frac{\sqrt{2m^3E}}{\hbar^3}\int_{\SSS^2}d^2\vomega \,\biggl|\widehat{\psi}\biggl(\frac{\sqrt{2mE}}{\hbar}\vomega\biggr)\biggr|^2 \,, 
\ee
where $\SSS^2=\{\vomega\in\RRR^3:|\vomega|=1\}$ is the unit sphere, $d^2\vomega $ the surface element ($\sin \theta \, d\theta \, d\varphi$ in spherical coordinates), and $\widehat{\psi}$ the Fourier transform of $\psi$. The standard deviation of this distribution is $\sigma_{E,\psi}$. Equivalently,
\be\label{sigmaEHfree}
\sigma_{E,\psi}^2 = \scp{\psi}{H_{\free}^2|\psi}-\scp{\psi}{H_{\free}|\psi}^2\,.
\ee

The following observation about $\psi\in C_0^\infty(\R\setminus\bou)$ is relevant here: Consider the three different viewpoints of taking the Hamiltonian to be $H$ (and regarding $\psi$ as an element of $L^2(\R)$), or of taking the Hamiltonian to be $H_{\free}$ (and regarding $\psi$ as an element of $L^2(\RRR^3)$), or of taking the Hamiltonian to be $H_{\Dir}=-(\hbar^2/2m)\nabla^2$ with Dirichlet boundary conditions on $\bou$ (and regarding $\psi$ as an element of $L^2(\R)$ again). The latter two viewpoints lead to completely different probability distributions over the energy axis (one is continuous, the other discrete for bounded $\R$), but both have the same mean and variance, as the mean is given by
\be
\scp{\psi}{H_{\free}|\psi}=\scp{\psi}{H_{\Dir}|\psi}
\ee
(because for a function not touching the boundary, both $H_{\free}$ and $H_{\Dir}$ are given by $-(\hbar^2/2m)\nabla^2$), and the variance by
\be
\scp{\psi}{H_{\free}^2|\psi}-\scp{\psi}{H_{\free}|\psi}^2=\scp{\psi}{H_{\Dir}^2|\psi}-\scp{\psi}{H_{\Dir}|\psi}^2
\ee
(because both $H_{\free}^2$ and $H_{\Dir}^2$ are $(\hbar^4/4m^2)\nabla^4$). Furthermore, 
the standard formulas for the mean and variance, $\scp{\psi}{H|\psi}$ and $\scp{\psi}{H^2|\psi}-\scp{\psi}{H|\psi}^2$, applied to the non-self-adjoint $H$, still yield the same values, for $\psi\in C_0^\infty (\R\setminus\bou)$, as $H_{\free}$ and $H_{\Dir}$. That is why $\sigma_{E,\psi}$ can still be expressed in terms of $H$ as
\be\label{sigmaH}
\sigma_{E,\psi}^2 = \scp{\psi}{H^2|\psi}-\scp{\psi}{H|\psi}^2\,.
\ee

As another remark, I mention that our proofs of \eqref{TE} and \eqref{TEp} remain valid if a bounded, smooth potential $V:\RRR^3\to\RRR$ with bounded derivatives is added to the Schr\"odinger equation \eqref{Schr}, and $\sigma_{E,\psi}$ is understood as the standard deviation for $-(\hbar^2/2m)\nabla^2+V$, in agreement with \eqref{sigmaH}. While for the uncertainty relation \eqref{QP} between position and momentum, the choice of potential is irrelevant, the relation \eqref{TE} between time and energy is affected by the choice of $V$ in two ways: first, because $V$ is part of the meaning of the energy $E$, and second, because the choice of $V$ affects the time evolution and thus the probability distribution of $T$.

\section{Proof of \eqref{TE}}
\label{sec:pfTE}

We first focus on the case that $p=\prob(T<\infty)=1$ for all $\psi$. We formulate our result for $V=0$. For definiteness, we take $\R$ to be open, $\R=\R\setminus \bou$.

\begin{thm}
Suppose that $\R\subset \RRR^3$ is open with a boundary $\bou$ that is locally Lipschitz and piecewise $C^1$, and such that
\be\label{WtWt0}
\lim_{t\to\infty} W_t^*W_t = 0
\ee
(or, equivalently, $p=1$ for every $\psi\in L^2(\R)$).
For every $\psi\in D(H^2)$ with $\|\psi\|=1$ and $\psi$ and $H\psi$ vanishing on $\bou$, in particular for $\psi\in C_0^\infty(\R)$ with $\|\psi\|=1$, \eqref{TE} is true with \eqref{sigmaT} and \eqref{sigmaEHfree}.
\end{thm}

\begin{proof}
Theorem 1 of \cite{detect-thm} guarantees that the operators $W_t$ exist and the distribution of $T$ is well defined. 
Due to \eqref{WtWt0},
\begin{subequations}\label{Jdef}
\be
J: L^2(\R,d^3\vx) \to L^2([0,\infty)\times \bou, dt\, d^2\vx)\,,
\ee 
defined by
\be
J\psi(t,\vx) = \sqrt{\hbar\kappa/m} \,(W_t\, \psi)(\vx)
\ee 
\end{subequations}
on $\psi$ from the domain of $H$ and by continuation on other $\psi$s from $L^2(\R)$ \cite[Sec.~5]{detect-thm}, is a unitary isomorphism between $\Hilbert=L^2(\R)$ and a subspace of $\Hilbert_{>}=L^2([0,\infty)\times\bou)$; the natural PVM on $\Hilbert_{>}$ is the Naimark dilation \cite{Nai40} of the joint POVM for detection time $T$ and detection location $\vX$; $T$ corresponds to a self-adjoint operator $\hat T = t$ on $\Hilbert_{>}$, where $t$ means multiplication by the variable $t$, and the probability distribution of $T$ for $\psi\in\Hilbert$ with $\|\psi\|=1$ is exactly the distribution associated with $\hat T$ and $J\psi$. The time evolution $W_t$ gets mapped by $J$ to the shift $S_t$ on $\Hilbert_{>}$, $S_t \phi(s,\vx)=\phi(s+t,\vx)$ (by virtue of the semigroup property):
\be\label{JW}
JW_t=S_tJ\,.
\ee
Correspondingly, $H$ gets mapped by $J$ to the generator $H_{>}$ of the semigroup $(S_t)_{t\geq 0}$, i.e., $H_{>}= i\hbar\partial/\partial t$, which is no longer self-adjoint because of the boundary at $t=0$. 

We will also regard $\Hilbert_{>}$ as a subspace of $\widetilde\Hilbert= L^2(\RRR\times \bou)$ (by setting $\phi(t,\vx)=0$ for $t<0$) with $\widetilde S_t$ the shift defined for positive or negative $t$. Note that $\widetilde S_t$ does not agree with $S_t$ even for positive $t$; $\Hilbert_{>}$ is not invariant under $\widetilde S_t$ for $t>0$, and $S_t f \neq \widetilde S_t f$ for $f\in\Hilbert_{>}\subset \widetilde \Hilbert$. Rather, for such $f$, $S_t f=P_{>}\widetilde{S}_t f$ with $P_{>}$ the projection $\widetilde{\Hilbert}\to \Hilbert_{>}$. The group $S_t$ is generated by the operator $\widetilde H = i\hbar\partial/\partial t$, whose formula looks the same as that of $H_{>}$, but the domain of $\widetilde H$ is different from that of $H_{>}$, and $\widetilde H$ is self-adjoint whereas $H_{>}$ is not. The domain of $\widetilde H$ is $H^1(\RRR,L^2(\bou))$, that of $H_{>}$ is $P_{>}H^1(\RRR,L^2(\bou))$.

The multiplication by $t$ also defines a self-adjoint operator $\widetilde T$ on $\widetilde\Hilbert$, and for $f\in \Hilbert_{>}\subset \widetilde\Hilbert$, $\hat T f = \widetilde T f$. Now on $\widetilde\Hilbert$, there is an uncertainty relation between $\widetilde T$ and $\widetilde H$: For any $\phi\in\widetilde \Hilbert$ with $\|\phi\|=1$,
\be\label{uncertaintywidetilde}
\sigma_{\widetilde T,\phi} \, \sigma_{\widetilde H,\phi} \geq \frac{\hbar}{2}\,.
\ee
This is simply the uncertainty relation \eqref{QP} for position and momentum because $t$ is now one of the variables in $\psi$, analogous to position, and $\widetilde H$ is the derivative relative to $t$, analogous to momentum.

Here, we need to comment on the mathematical conditions of the validity of \eqref{QP}. The question is whether, for \eqref{QP} to be valid, $\psi$ needs to lie in the domain of $Q$ and that of $P$, or perhaps even in the domain of $QP$ and that of $PQ$. The answer is that \eqref{QP}, when understood appropriately, is valid for every $\psi\in L^2(\RRR)$ with $\|\psi\|=1$ (whereas this is not true in general for \eqref{AB} \cite[Chap.~12]{Hall}). Let us explain. Let $A$ be a self-adjoint operator in the Hilbert space $\Hilbert$ and $\psi\in\Hilbert$ with $\|\psi\|=1$. Then $\psi$ defines a probability distribution on the spectrum of $A$, the Born distribution $\scp{\psi}{P(\cdot)|\psi}$ with $P$ the spectral PVM of $A$, $A=\int_{\RRR}P(d\alpha) \, \alpha$. The Born distribution has finite expectation and variance if and only if $\psi$ lies in the domain of $A$. In that case, the variance is given by
\be\label{sigmaA}
\sigma_{A,\psi}^2=\|A\psi\|^2 - \scp{\psi}{A\psi}^2\,.
\ee
Thus, for every $\psi\in L^2(\RRR)$ with $\|\psi\|=1$ that lies in both the domain of $Q$ and that of $P$, $\sigma_{Q,\psi}$ and $\sigma_{P,\psi}$ are well defined and finite. For those $\psi$, \eqref{QP} can be proved \cite[Thm.~12.7]{Hall}. However, since for $\psi$ not in the domain of $Q$, $\sigma_{Q,\psi}=\infty$, and since there is no unit vector in $L^2(\RRR)$ for which either $\sigma_Q$ or $\sigma_P$ vanishes, the left-hand side of \eqref{QP} is infinite (and the relation trivially true) for every $\psi$ not in both the domain of $Q$ and that of $P$. Thus, \eqref{QP} is plainly and cleanly valid for all wave functions.

The remainder of the reasoning is about the implications of the relation \eqref{uncertaintywidetilde} in $\widetilde\Hilbert$ for $\psi$ that vanishes along with $H\psi$ on $\bou$, such as $\psi\in C_0^\infty(\R)$. Since $\psi$ lies in the domain of $H$, $J\psi$ lies in the domain of $H_{>}$, and $H_{>}J\psi=JH\psi$. Moreover, $J\psi(t=0,\vx)=0$ (because $\psi$ vanishes on $\bou$), and $(\partial/\partial t)J\psi(t=0,\vx)=0$ in $L^2(\bou)$ (because $H\psi$ still lies in the domain of $H$ and thus has a trace on $\bou$, which in fact is zero). Therefore, $J\psi$ also lies in the domain of $\widetilde{H}$, and
\be
\widetilde{H}J\psi = H_{>}J\psi = JH\psi\,.
\ee
Thus, by \eqref{sigmaA},
\begin{subequations}
\begin{align}
\sigma_{\widetilde{H},J\psi}^2 
&= \|\widetilde{H}J\psi\|^2-\scp{J\psi}{\widetilde{H}J\psi}^2\\
&= \|JH\psi\|^2-\scp{J\psi}{JH\psi}^2\\[1.5mm]
&= \|H\psi\|^2-\scp{\psi}{H\psi}^2\\[1mm]
&=\sigma_{E,\psi}^2\,,
\end{align}
\end{subequations}
using that $J$ is a unitary isomorphism and \eqref{sigmaH}. 
On the other hand,
\begin{align}
\sigma_{\widetilde{T},J\psi}^2 
&=\scp{J\psi}{\widetilde{T}^2|J\psi}-\scp{J\psi}{\widetilde{T}|J\psi}^2\\
&=\scp{J\psi}{\hat{T}^2|J\psi}-\scp{J\psi}{\hat{T}|J\psi}^2\\[2mm]
&=\sigma_{T,\psi}^2\,,
\end{align}
which completes the proof of \eqref{TE}.
\end{proof}

Let us turn again to the question of the sharp constant in \eqref{TE}. It is well known that the constant in \eqref{QP} is sharp and equality holds for suitable Gaussian packets. As a consequence, $\sigma_{\widetilde T} \, \sigma_{\widetilde H}=\hbar/2$ for some states in $\widetilde\Hilbert$, and can presumably come arbitrarily close to $\hbar/2$ for suitable states in $\Hilbert_{>}$. The difficulty with making corresponding statements about $\sigma_T \, \sigma_E$ for states in $L^2(\R)$ is to characterize the functions in the range of $J$ and to control how close they can come to Gaussian packets in $\widetilde{\Hilbert}$.

\section{Proof of \eqref{TEp}}
\label{sec:pfTEp}

In \eqref{TEp}, we mean by $\sigma_{T,\psi}$ the standard deviation of $T$, conditional on $T<\infty$:
\be\label{sigmaTp}
\sigma_{T,\psi}^2 = \mathrm{Var}(T|T<\infty)\,.
\ee

\begin{thm}
Suppose that $\R\subset \RRR^3$ is open with a boundary $\bou$ that is locally Lipschitz and piecewise $C^1$. For every $\psi\in D(H^2)$ with $\|\psi\|=1$ and $\psi$ and $H\psi$ vanishing on $\bou$, in particular for $\psi\in C_0^\infty(\R)$ with $\|\psi\|=1$, \eqref{TEp} is true with \eqref{sigmaTp} and \eqref{sigmaEHfree}.
\end{thm}

\begin{proof}
The mapping $J$ defined by \eqref{Jdef} is now not unitary but still a contraction, 
\be\label{Jcontract}
\|J\phi\| \leq \|\phi\| \quad \forall \phi\in L^2(\R)\,.
\ee
Since $\|J\psi\|=\sqrt{p}$, $J\psi/\sqrt{p}$ is a unit vector.
As before,
\be
S_t J=J W_t\,,~~\text{so}~~
H_{>} J = J H\,,
\ee
and the assumption on $\psi$ implies again that $J\psi$ lies in the domain of $\widetilde H$ and
\be
\widetilde H J\psi=JH\psi\,.
\ee
By the Kennard relation for $t$ and $i\partial/\partial t$,
\begin{subequations}\label{Hp}
\begin{align}
\frac{\hbar^2}{4\sigma^2_{T,J\psi/\sqrt{p}}} 
&\leq \sigma^2_{\widetilde{H},J\psi/\sqrt{p}}\\
&= \frac{1}{p} \|\widetilde{H}J\psi\|^2 - \frac{1}{p^2} \scp{J\psi}{\widetilde{H}J\psi}^2\\
&\leq \frac{1}{p} \bigl\|\widetilde{H}J\psi\bigr\|^2\\
&= \frac{1}{p} \bigl\| J H \psi\bigr\|^2\\
&\stackrel{\eqref{Jcontract}}{\leq} \frac{1}{p} \bigl\| H \psi\bigr\|^2\\
&= \frac{1}{p} \bigl\| H_\free \psi\bigr\|^2\,.
\end{align}
\end{subequations}
Since for any self-adjoint operator $A$, the variance $\sigma^2_A$ does not change when adding a constant $c\in\RRR$ to $A$, $\sigma^2_{A+cI}=\sigma^2_A$, we can replace $\widetilde{H}$ in \eqref{Hp} by $\widetilde{H}+E_0I$; thus, for any $E_0\in \RRR$,
\be\label{Hfreep}
\frac{\hbar^2}{4\sigma^2_{T,J\psi/\sqrt{p}}} 
\leq \frac{1}{p} \bigl\| (H_\free+E_0I) \psi\bigr\|^2\,.
\ee
Since for any real random variable $X$, its variance can be characterized as
\be
\text{Var} \, X = \inf_{c\in\RRR} \EEE\bigl[ (X-c)^2 \bigr]\,,
\ee
we have that for any observable $A$,
\be
\sigma_{A,\psi}^2 = \inf_{c\in\RRR} \scp{\psi}{(A-cI)^2|\psi}\,,
\ee
in particular
\be
\sigma_{H_\free,\psi}^2= \inf_{E_0\in\RRR} \bigl\| (H_\free+E_0I) \psi\bigr\|^2\,.
\ee
Thus, with \eqref{Hfreep},
\be
\frac{\hbar^2}{4\sigma^2_{T,J\psi/\sqrt{p}}} 
\leq \frac{1}{p} \sigma_{H_\free,\psi}^2\,,
\ee
which is just a different notation for the desired relation \eqref{TEp}, so the proof is complete.
\end{proof}

\bigskip

\noindent{\it Acknowledgment.} I thank Stefan Teufel for helpful discussion.

\bigskip

\noindent{\it Conflict of interest statement.} On behalf of all authors, the corresponding author states that there is no conflict of interest. 

\bigskip

\noindent{\it Data availability statement.} Not applicable.


\begin{thebibliography}{18.}

\bibitem{Ken27} E.H. Kennard:
	Zur Quantenmechanik einfacher Bewegungstypen.
	{\it Zeitschrift f\"ur Physik} {\bf 44(4-5)}: 326 (1927) 

\bibitem{Rob} H.P. Robertson:
	The Uncertainty Principle.
	{\it Physical Review} {\bf 34}: 163--164 (1929)

\bibitem{Schr} E. Schr\"odinger:
	Zum Heisenbergschen Unsch\"arfeprinzip.
	{\it Sitzungsberichte der Preussischen Akademie der Wissenschaften, 
	physikalisch-mathematische Klasse} {\bf 14}: 296--303 (1930)

\bibitem{AB61} Y. Aharonov and D. Bohm:
	Time in the Quantum Theory and the Uncertainty Relation for Time and Energy.
	{\it Physical Review} {\bf 122}: 1649 (1961)

\bibitem{Busch08} P. Busch:
	The Time--Energy Uncertainty Relation.
	Pages 73--106 in J.G. Muga, R. Sala Mayato, and \'I.L. Egusquiza (editors),
	{\it Time in Quantum Mechanics, Vol.~1, Second Edition}, 
	Lecture Notes in Physics {\bf 734},
	Berlin: Springer-Verlag (2008)
	\url{http://arxiv.org/abs/quant-ph/0105049v3}

\bibitem{KRSW12} J. Kiukas, A. Ruschhaupt, P.O. Schmidt, and R.F. Werner:
	Exact Energy--Time Uncertainty Relation for Arrival Time by Absorption.
	{\it Journal of Physics A: Mathematical and Theoretical} {\bf 45}: 185301 (2012)
	\url{http://arxiv.org/abs/1109.5087}

\bibitem{DD15} V.V. Dodonov and A.V. Dodonov:
	Energy--time and frequency--time uncertainty relations: exact inequalities.
	{\it Physica Scripta} {\bf 90}: 074049 (2015)
	\url{http://arxiv.org/abs/1504.00862}

\bibitem{Wer87} R. Werner:
	Arrival time observables in quantum mechanics.
	\textit{Annales de l'Institute Henri Poincar\'e, section A} \textbf{47} 429--449 (1987)

\bibitem{detect-rule} R. Tumulka:
	Distribution of the Time at Which an Ideal Detector Clicks.
	To appear in {\it Annals of Physics} (2022)
	\url{http://arxiv.org/abs/1601.03715}
	
\bibitem{detect-several} R. Tumulka:
	Detection Time Distribution for Several Quantum Particles.
	\url{http://arxiv.org/abs/1601.03871}

\bibitem{detect-dirac} R. Tumulka:
	Detection Time Distribution for the Dirac Equation.
	\url{http://arxiv.org/abs/1601.04571}

\bibitem{detect-imaginary} R. Tumulka:
	Absorbing Boundary Condition as Limiting Case of Imaginary Potentials.
	\url{http://arxiv.org/abs/1911.12730}
	
\bibitem{detect-thm} S. Teufel and R. Tumulka:
	Existence of Schr\"odinger Evolution with Absorbing Boundary Condition.
	To appear in {\it Quantum Studies: Mathematics and Foundations} (2022)
	\url{http://arxiv.org/abs/1912.12057}

\bibitem{All69b} G.R. Allcock:
	The time of arrival in quantum mechanics II. The individual measurement.
	{\it Annals of Physics} {\bf 53}: 286--310 (1969)

\bibitem{Nai40} M.A. Naimark:
	Spectral functions of a symmetric operator.
	{\it Izvestiya Akademii Nauk SSSR. Seriya Matematicheskaya}
	{\bf 4(3)}: 277--318 (1940)

\bibitem{Lud85} G. Ludwig:
	\textit{Foundations of quantum mechanics, Vol.~I and II.}
	Berlin: Springer-Verlag (1983, 1985)

\bibitem{povm} R. Tumulka:
	POVM.
	Pages 479--483 in D. Greenberger, K. Hentschel, and F. Weinert (editors):
	\textit{Compendium of Quantum Physics,}
	Berlin: Springer-Verlag (2009)	

\bibitem{Hall} B.C. Hall:
	{\it Quantum Theory for Mathematicians.}
	New York: Springer (2013)

\end{thebibliography}
\end{document}